\documentclass[11pt]{article}    

\usepackage{cite} 
\usepackage{url}  
\usepackage{ifthen}  
\usepackage{multicol}   
\usepackage{setspace}

\usepackage{mathptmx}	
\usepackage{graphicx}
\usepackage{subfigure} 
\usepackage{amsmath, amsfonts, amssymb, amsthm}
\usepackage{mathtools}

\usepackage[font=footnotesize,width=.85\textwidth,labelfont=bf]{caption}
\usepackage[mathcal]{euscript}
\usepackage{mathptmx} 

\usepackage{authblk}

\newcommand{\lca}{\ensuremath{\operatorname{lca}}}

\newtheorem{theorem}{Theorem}
\newtheorem{lemma}[theorem]{Lemma}
\newtheorem{definition}[theorem]{Definition}
\newtheorem{corollary}[theorem]{Corollary}

\usepackage{color}

\usepackage{fixltx2e}

\begin{document}
\title{From event labeled gene trees to species trees}


\author[3,5]{Maribel Hernandez-Rosales}
\author[1,2]{Marc Hellmuth}
\author[8]{Nicolas Wieseke}
\author[4]{Katharina T.~Huber}
\author[4]{Vincent Moulton}
\author[3,5,6,7]{Peter F.\ Stadler}

\affil[1]{\footnotesize Dpt.\ of Mathematics and Computer Science, University of Greifswald, Walther-
  Rathenau-Strasse 47, D-17487 Greifswald, Germany }
\affil[2]{Saarland University, Center for Bioinformatics, Building E 2.1, P.O.\ Box 151150, D-66041 Saarbr{\"u}cken, Germany }
\affil[3]{Bioinformatics Group, Department of Computer Science; and
		    Interdisciplinary Center of Bioinformatics, University of Leipzig, \\
			 H{\"a}rtelstra{\ss}e 16-18, D-04107 Leipzig, Germany}
\affil[4]{School of Computing Sciences, University of East Anglia,
		    Norwich, NR4 7TJ, UK.\ }
\affil[5]{Max-Planck-Institute for Mathematics in the Sciences, \\
  Inselstra{\ss}e 22, D-04103 Leipzig, Germany}
\affil[6]{Inst.\ f.\ Theoretical Chemistry, University of Vienna, \\
  W{\"a}hringerstra{\ss}e 17, A-1090 Wien, Austria}
\affil[7]{Santa Fe Institute, 1399 Hyde Park Rd., Santa Fe, USA} 
\affil[8]{Parallel Computing and Complex Systems Group \\
  Department of Computer Science, 
  Leipzig University \\
  Augustusplatz 10, 04109, Leipzig, Germany}
\date{}
\normalsize

\maketitle

\abstract{ \noindent
 \paragraph*{Background:} 
  Tree reconciliation problems have long been studied in phylogenetics.  A
  particular variant of the reconciliation problem for a gene tree $T$ and
  a species tree $S$ assumes that for each interior vertex $x$ of $T$ it is
  known whether $x$ represents a speciation or a duplication. This problem
  appears in the context of analyzing orthology data.
  \paragraph*{Results:} 
  We show that $S$ is a species tree for $T$ if and only if $S$ displays
  all rooted triples of $T$ that have three distinct species as their
  leaves and are rooted in a speciation vertex. A valid reconciliation map
  can then be found in polynomial time. Simulated data shows that the
  event-labeled gene trees convey a large amount of information on
  underlying species trees, even for a large percentage of losses.
  \paragraph*{Conclusions:} 
  The knowledge of event labels in a gene tree strongly constrains the
  possible species tree and, for a given species tree, also the possible
  reconciliation maps. Nevertheless, many degrees of freedom remain in the
  space of feasible solutions. In order to disambiguate the alternative
  solutions additional external constraints as well as optimization criteria
  could be employed.
}

\sloppy

\section*{Background}

The reconstruction of the evolutionary history of a gene family is
necessarily based on at least three interrelated types of information. The
true phylogeny of the investigated species is required as a scaffold with
which the associated gene tree must be reconcilable.
Orthology or paralogy of genes found in different species determines
whether an internal vertex in the gene tree corresponds to a duplication or
a speciation event. Speciation events, in turn, are reflected in the
species tree.

The reconciliation of gene and species trees is a widely studied
problem \cite{Guigo1996,Page1997,Arvestad2003,Bonizzoni2005,Gorecki2006,%
Hahn2007,Bansal2008,Chauve2008,Burleigh2009,Larget2010}. In most
practical applications, however, neither the gene tree nor the species tree
can be determined unambiguously.

Although orthology information is often derived from the reconciliation of
a gene tree with a species tree (cf.\ e.g.\ \texttt{TreeFam} \cite{Li:06},
\texttt{PhyOP} \cite{Goodstadt:06}, \texttt{PHOG} \cite{Datta:09},
\texttt{EnsemblCompara GeneTrees} \cite{Hubbard:07}, and \texttt{MetaPhOrs}
\cite{Pryszcz:11}), recent benchmarks studies \cite{Altenhoff:09} have
shown that orthology can also be inferred at similar levels of accuracy
without the need to construct trees by means of clustering-based approaches
such as \texttt{OrthoMCL} \cite{Li:03}, the algorithms underlying the
\texttt{COG} database \cite{Tatusov:00,Wheeler:08}, \texttt{InParanoid}
\cite{Berglund:08}, or \texttt{ProteinOrtho} \cite{Lechner:11a}.  In
\cite{Hellmuth:12d} we have therefore addressed the question: How much
information about the gene tree, the species tree, and their reconciliation
is already contained in the orthology relation between genes?

According to Fitch's definition \cite{Fitch2000}, two genes are
(co-)orthologous if their last common ancestor in the gene tree represents
a speciation event.  Otherwise, i.e., when their last common ancestor is a
duplication event, they are paralogs. The orthology relation on a set of
genes is therefore determined by the gene tree $T$ and an ``event
labeling'' that identifies each interior vertex of $T$ as either a
duplication or a speciation event. (We disregard here additional types of
events such as horizontal transfer and refer to \cite{Hellmuth:12d} for
details on how such extensions might be incorporated into the mathematical
framework.)  One of the main results of \cite{Hellmuth:12d}, which relies
on the theory of symbolic ultrametrics developed in \cite{Boeckner:98}, is
the following: A relation on a set of genes is an orthology relation (i.e.,
it derives from some event-labeled gene tree) if and only if it is a
cograph (for several equivalent characterizations of cographs see
\cite{Brandstaedt:99}).  Note that the cograph does not contain the full
information on the event-labeled gene tree. Instead the cograph is
equivalent to the gene tree's homomorphic image obtained by collapsing
adjacent events of the same type \cite{Hellmuth:12d}. The orthology
relation thus places strong and easily interpretable constraints on the
gene tree.

This observation suggests that a viable approach to reconstructing
histories of large gene families may start from an empirically determined
orthology relation, which can be directly adjusted to conform to the
requirement of being a cograph. The result is then equivalent to an
(usually incompletely resolved) event-labeled gene tree, which might be
refined or used as constraint in the inference of a fully resolved gene
tree.  In this contribution we are concerned with the next conceptual step:
the derivation of a species tree from an event-labeled gene tree. As we
shall see below, this problem is much simpler than the full tree
reconciliation problem.  Technically, we will approach this problem by
reducing the reconciliation map from gene tree to species tree to rooted
triples of genes residing in three distinct species. This is related to an
approach that was developed in \cite{Chauve:09} for addressing the full
tree reconciliation problem.

\section*{Methods}

\subsection*{Definitions and Notation}

\subsubsection*{Phylogenetic Trees} 

A \emph{phylogenetic tree $T$ (on $L$)} is a rooted tree $T=(V,E)$, with
leaf set $L\subseteq V$, set of directed edges $E$, and set of interior
vertices $V^0=V\setminus L$ that does not contain any vertices with in- and
outdegree one and whose root $\rho_T\in V$ has indegree zero. In order to
avoid uninteresting trivial cases, we assume that $|L|\ge 3$. The ancestor
relation $\preceq_T$ on $V$ is the partial order defined, for all $x,y\in V$,
by $x \preceq_T y$ whenever $y$ lies on the path from $x$ to the root. If
there is no danger of ambiguity, we will write $x\preceq y$ rather than $x
\preceq_T y$.  Furthermore, we write $x \prec y$ to mean $x \preceq y$ and
$x\ne y$. For $x\in V$, we write $L(x):=\{ y\in L| y\preceq x\}$ for the
set of leaves in the subtree $T(x)$ of $T$ rooted in $x$. Thus,
$L(\rho_T)=L$ and $T(\rho_T)=T$. For $x,y\in V$ such that $x$ and $y$ are
joined by an edge $e\in E$ we write $e=[y,x]$ if $x\prec y$. Two
phylogenetic trees $T=(V,E)$ and $T'=(V',E')$ on $L$ are said to be
\emph{equivalent} if there exists a bijection from $V$ to $V'$ that is the
identity on $L$, maps $\rho_T$ to $\rho_{T'}$, and extends to a graph
isomorphism between $T$ and $T'$.  A \emph{refinement} of a phylogenetic
tree $T$ on $L$ is a phylogenetic tree $T'$ on $L$ such that $T$ can be
obtained from $T'$ by collapsing edges (see e.g.\ \cite{Semple:book}).

Suppose for the remainder of this section that $T=(V,E)$ is a phylogenetic
tree on $L$ with root $\rho_T$. For a non-empty subset of leaves
$A\subseteq L$, we define $\lca_T(A)$, or the \emph{most recent common
  ancestor of $A$}, to be the unique vertex in $T$ that is the greatest
lower bound of $A$ under the partial order $\preceq_T$.  In case $A=\{x,y
\}$, we put $\lca_T(x,y):=\lca_T(\{x,y\})$ and if $A=\{x,y,z \}$, we put
$\lca_T(x,y,z):=\lca_T(\{x,y,z\})$.  For later reference, we have, for all
$x\in V$, that $x=\lca_T(L(x))$.  Let $L'\subseteq L$ be a subset of
$|L'|\ge 2$ leaves of $T$. We denote by $T(L')=T(\lca_T(L'))$ the (rooted)
subtree of $T$ with root $\lca_T(L')$. Note that $T(L')$ may have leaves
that are not contained in $L'$. The \emph{restriction} $T|_{L'}$ of $T$ to
$L'$ is the phylogenetic tree with leaf set $L'$ obtained from $T$ by first 
forming the minimal spanning tree in $T$ with leaf set $L'$ and then
by suppressing all vertices of degree two with the exception of $\rho_T$ if
$\rho_T$ is a vertex of that tree. A phylogenetic tree $T'$ on some subset
$L'\subseteq L$ is said to be \emph{displayed} by $T$ (or equivalently that
$T$ \emph{displays} $T'$) if $T'$ is equivalent with $T|_{L'}$. A set
$\mathcal{T}$ of phylogenetic trees $T$ each with leaf set $L_T$ is called
\emph{consistent} if $\mathcal{T}=\emptyset$ or there is a phylogenetic
tree $T$ on $L=\bigcup_{T\in \mathcal T} L_T$ that \emph{displays}
$\mathcal T$, that is, $T$ displays every tree contained in
$\mathcal{T}$. Note that a consistent set of phylogenetic trees is
sometimes also called compatible (see e.g.\ \cite{Semple:book}).

It will be convenient for our discussion below to extend the ancestor
relation $\preceq_T$ on $V$ to the union of the edge and vertex sets of
$T$. More precisely, for the directed edge $e=[u,v]\in E$ we put $x
\prec_T e$ if and onfly if $x\preceq_T v$ and $e \prec_T x$ if and 
only if $u\preceq_E x$. For edges $e=[u,v]$ and $f=[a,b]$ in T we put 
$e\preceq f$ if and only if $v \preceq b$.

\subsubsection*{Rooted triples} 

Rooted triples are phylogenetic trees on three leaves with precisely two
interior vertices. Sometimes also called rooted triplets \cite{Dress:book} 
they constitute an important concept in the context of supertree reconstruction
\cite{Semple:book,Bininda:book} and will also play a major role here.
Suppose $L=\{x,y,z\}$.  Then we denote by $((x,y),z)$ the triple $r$ with
leaf set $L$ for which the path from $x$ to $y$ does not intersect the path
from $z$ to the root $\rho_r$ and thus, having $\lca_r(x,y)\prec \lca_r(x,y,z)$.
For $T$ a phylogenetic tree, we denote by $\mathfrak{R}(T)$ the set of all
triples that are displayed by $T$.  

Clearly, a set $\mathcal{R}$ of triples is
\emph{consistent} if there is a phylogenetic tree $T$ on 
$X=\bigcup_{r\in \mathcal{R}} L(\rho_r)$ such that 
$\mathcal{R} \subseteq \mathcal{R}(T)$. 
Not all sets of triples are consistent of course. Given a triple set
$\mathcal{R}$ there is a polynomial-time algorithm, referred to in
\cite{Semple:book} as \texttt{BUILD}, that either constructs a phylogenetic
tree $T$ that displays $\mathcal{R}$ or that recognizes that $\mathcal{R}$
is \emph{inconsistent}, that is, not consistent \cite{Aho:81}.  Various
practical implementations have been described starting with \cite{Aho:81},
improved variants are discussed in \cite{Henzinger:99,Jansson:05}.

The problem of determining a maximum consistent subset $\mathcal{R}'$ of an
inconsistent set of triples, on the other hand is NP-hard and also
APX-hard, see \cite{Byrka:10a,vanIersel:09} and the references therein.  We
refer to \cite{Byrka:10} for an overview on the available practical
approaches and further theoretical results.

The \texttt{BUILD} algorithm, furthermore, does not necessarily generate
for a given triple set $\mathcal R$ a minimal phylogenetic tree $T$ that
displays $\mathcal{R} $, i.e., $T$ may resolve multifurcations in an
arbitrary way that is not implied by any of the triples in
$\mathcal{R}$. However, the tree generated by \texttt{BUILD} is
minor-minimal, i.e., if $T'$ is obtained from $T$ by contracting an edge,
$T'$ does not display $\mathcal{R}$ anymore. The trees produced by
\texttt{BUILD} do not necessarily have the minimum number of internal
vertices. Thus, depending on $\mathcal{R}$, not all trees consistent with
$\mathcal{R}$ can be obtained from \texttt{BUILD}.  Semple \cite{Semple:03}
gives an algorithm that produces all minor-minimal trees consistent with
$\mathcal{R}$. It requires only polynomial time for each of the possibly
exponentially many minor-minimal trees.  The problem of constructing a tree
consistent with $\mathcal{R}$ and minimizing the number of interior
vertices in NP-hard and hard to approximate \cite{Jansson:12}.

\subsection*{Event Labeling, Species Labeling, and Reconciliation Map}

A gene tree $T$ arises through a series of events along a species tree
$S$. We consider both $T$ and $S$ as phylogenetic trees with leaf sets $L$
(the set of genes) and $B$ (the set of species), respectively. We assume
that $|L|\ge 3$ and $|B|\ge 1$. We consider only gene duplications and gene
losses, which take place between speciation events, i.e., along the edges
of $S$. Speciation events are modeled by transmitting the gene content of
an ancestral lineage to each of its daughter lineages.

\begin{figure}[tbp]
\begin{center}
    \includegraphics[width=0.9\textwidth]{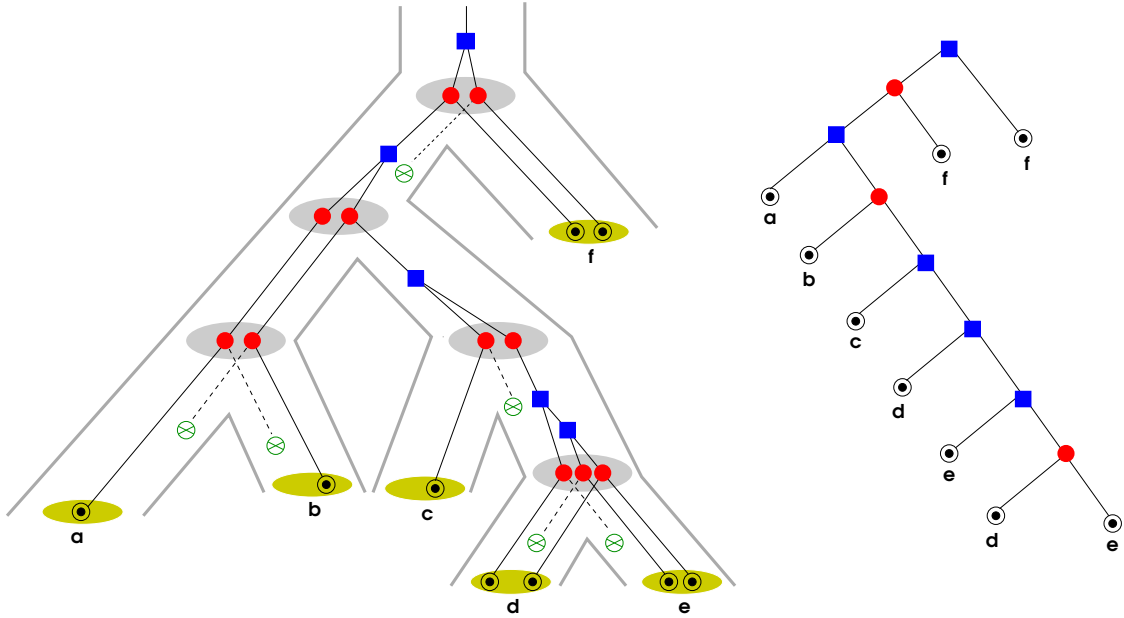}
\end{center}
\caption{ \textbf{Left:} Example of an evolutionary scenario showing
  the evolution of a gene family.  The corresponding true gene tree
  $\hat{T}$ appears embedded in the true species tree $\hat{S}$.  The
  map $\hat\mu$ is implicitly given by drawing the species tree
  superimposed on the gene tree. In particular, the speciation
  vertices in the gene tree (red circuits) are mapped to the vertices
  of the species tree (gray ovals) and the duplication vertices (blue
  squares) to the edges of the species tree. Gene losses are
  represented with ``$\otimes$'' (mapping to edges in $\hat{S}$). 
  The observable species
  $a,b,\ldots,f$ are the leaves of the species tree (green ovals) and
  extand genes therein are labeled with ``$\odot$''.
  \newline
  \textbf{Right:} The corresponding gene tree $T$ with observed events
  from the left tree. Leaves are labeled with the corresponding
  species.}
\label{fig:true}
\end{figure}

The true evolutionary history of a single ancestral gene thus can be
thought of as a scenario such as the one depicted in
Fig.~\ref{fig:true}. Since we do not consider horizontal gene transfer or
lineage sorting in this contribution, an evolutionary scenario consists of
four components: (1) A true (binary) gene tree $\hat{T}$, (2) a true (binary) species tree
$\hat{S}$, (3) an assignment of an event type (i.e., speciation $\bullet$,
duplication $\square$, loss $\otimes$, or observable (extant) gene $\odot$)
to each interior vertex and leaf of $\hat{T}$, and (4) a map $\mu$ assigning
every vertex of $\hat{T}$ to a vertex or edge of $\hat{S}$ in such a way
that (a) the ancestor order of $\hat{T}$ is preserved, (b) a vertex of
$\hat{T}$ is mapped to an interior vertex of $\hat{S}$ if and only if it is
of type speciation, (c) extant genes of $\hat{T}$ are mapped to leaves of
$\hat{S}$.\footnote{Alternatively, one could define $\hat T$ and $\hat S$
  to be metric graphs (i.e., comprising edges that are real intervals glued
  together at the vertices) with a distance function that measures
  evolutionary time. In this picture, $\hat\mu$ is a continuous map that
  preserves the temporal order and satisfied conditions (b) and (c).}

In order to allow $\hat\mu$ to map duplication vertices to a
time point before the last common ancestor of all species in $\hat{S}$,
we need to extend our definition of a species tree by adding an extra 
vertex and an extra edge ``above'' the last common ancestor of all
species. Note that strictly speaking
$\hat{S}$ is not a phylogenetic tree anymore. In case there is no danger of
confusion, we will from now on refer to a phylogenetic tree on $B$ with
this extra edge and vertex added as a species tree on $B$ and to $\rho_B$
as the root of $B$.  Also, we canonically extend our notions of a triple,
displaying, etc.\ to this new type of species tree.

The true gene tree $\hat{T}$ represents all extant as well as all
extinct genes, all duplication, and all speciation events. Not all of these
events are observable from extant genes data, however. In particular,
extinct genes cannot be observed. The observable part $T=T(V,E)$ of $\hat{T}$
is the restriction of $\hat{T}$ to the leaf set $L$ of extant genes, 
i.e., $T=\hat T|_{L}$. Hence, $T$ is still binary.

Furthermore, we can observe a map $\sigma:L\to B$ that assigns to each
extant gene the species in which it resides. Of course, for $x\in L$ we
have $\sigma(x)=\hat\mu(x)$. Here $B$ is the leaf set of the extant species
tree, i.e., $B=\sigma(L)$.  For ease of readability, we also put
$\sigma(T')=\{\sigma(x)\,:\, x\in L(y)\}$ for any subtree $T'$ of $T$ with
$T'=T(y)$ where $y\in V^0$. Alternatively, we will sometimes also write
$\sigma(y)$ instead of $\sigma(T(y))$.  Last but not least, for $Y\subseteq
L$, we put $\sigma(Y)=\{\sigma(y)\,:\, y\in Y\}$.

The observable part of the species tree  $S=(W,H)$ is the restriction 
$\hat S|_{B}$ of $\hat{S}$ to $B$. In order to
account for duplication events that occurred before the first speciation
event, the additional vertex $\rho_S\in W$ and the additional edge
$[\rho_S,\lca_S{B}]\in H$ must be part of $S$. 

The evolutionary scenario also implies an \emph{event labeling} map
$t:V\to\{\bullet,\square,\odot\}$ that assigns to each interior vertex $v$
of $T$ a value $t(v)$ indicating whether $v$ is a speciation event
($\bullet$) or a duplication event ($\square$). It is convenient to use the
special label $\odot$ for the leaves $x$ of $T$.  We write $(T,t)$ for the
event-labeled tree. We remark that $t$ was introduced as ``symbolic dating
map'' in \cite{Boeckner:98}.  It is called \emph{discriminating} if, for
all edges $\{u,v\}\in E$, we have $t(u) \neq t(v)$ in which case $(T,t)$ is
known to be in 1-1-correspondence to a cograph \cite{Hellmuth:12d}
(see e.g.\ \cite{HW:16book,HW:16a,HSW:16,HW:15,HLS+15} for further discussions and results 
 on tree-representable binary relations).
Note
that we will in general not require that $t$ is discriminating in this
contribution.  For $T=(V,E)$ a gene tree on $L$, $B$ a set of species,
and maps $t$ and $\sigma$ as specified above, we require however that
$\mu$ and $\sigma$ must satisfy the following compatibility property:
\begin{itemize}
\item[(C)] Let $z\in V$ be a speciation vertex, i.e., $t(z)=\bullet$, and let
  $T'$ and $T''$ be subtrees of $T$ rooted in two distinct children of $z$.  
Then  $\sigma(T')\cap \sigma(T'') = \emptyset$.
\end{itemize}
Note the we do not require the converse, i.e., from the disjointness of the
species sets $\sigma(T') $ and $\sigma(T'')$ we do \textbf{not} conclude
that their last common ancestor is a speciation vertex.

For $x,y\in L$ and $z=\lca_T(x,y)$ it immediately follows from condition
(C) that if $t(\lca_T(x,y))=\bullet$ then $\sigma(x)\ne\sigma(y)$ since, by
assumption, $x$ and $y$ are leaves in distinct subtrees below $z$.
Equivalently, two distinct genes $x\ne y$ in $L$ for which
$\sigma(x)=\sigma(y)$ holds, that is, they are contained in the same
species of $B$, must have originated from a duplication event, i.e.,
$t(\lca_T(x,y))=\square$. Thus we can regard $\sigma$ as a proper vertex
coloring of the cograph corresponding to $(T,t)$.

Let us now consider the properties of the restriction of $\hat\mu$ to the
observable parts $T$ of $\hat{T}$ and $S$ of $\hat{S}$.  Consider a
speciation vertex $x$ in $\hat{T}$. If $x$ has two children $y'$ and $y''$
so that $L(y')$ and $L(y'')$ are both non-empty then
$x=\lca_{\hat{T}}(z',z'')$ for all $z'\in L(y')$ and $z''\in L(y'')$ and
hence, $x = \lca_T (L(y')\cup L(y''))$.  In particular, $x$ is an observable
vertex in $T$.  Furthermore, we know that $\sigma(L(y'))\cap \sigma(L(y''))
= \emptyset$, and therefore, $\hat\mu (x) = \lca_S(\sigma (L(y')\cup
L(y''))$.  Considering all pairs of children with this property this can be
rephrased as $\hat\mu(x)=\lca_{\hat S}(\sigma(L(x)))$.  On the other hand,
if $x$ does not have at least two children with this property, and hence
the corresponding speciation vertex cannot be viewed as most recent common
ancestor of the set of its descendants in $S$, then $x$ is not a vertex in
the restriction $T=\hat T|_{L}$ of $\hat T$ to the set $L$ of the extant
genes. The restriction $\mu$ of $\hat\mu$ to the observable tree $T$
therefore satisfies the properties used below to define reconciliation
maps. 

\begin{definition} \label{def:mu} 
  Suppose that $B$ is a set of species,
  that $S=(W,H)$ is a phylogenetic tree on $B$, that $T=(V,E)$ is a gene
  tree with leaf set $L$ and that $\sigma:L\to B$ and $t:V\to
  \{\bullet,\square,\odot\} $ are the maps described above. Then we say
  that $S$ is a \emph{species tree for $(T,t,\sigma)$} if there is a map
  $\mu:V\to W\cup H$ such that, for all $x\in V$:
\begin{itemize}
\item[(i)]   If $t(x)=\odot$   then $\mu(x)=\sigma(x)$.
\item[(ii)]  If $t(x)=\bullet$ then $\mu(x)\in W\setminus B$.
\item[(iii)] If $t(x)=\square$ then $\mu(x)\in H$. 
\item[(iv)]  Let $x,y\in V$ with $x\prec_T y$. We distinguish two cases:
  \begin{enumerate}
  \item  If $t(x)=t(y)=\square$ then $\mu(x)\preceq_S \mu(y)$ in $S$.
  \item  If $t(x)=t(y)=\bullet$ or $t(x)\neq t(y)$ then 
    $\mu(x)\prec_S\mu(y)$ in $S$.
  \end{enumerate}
\item[(v)]   If $t(x)=\bullet$ then 
        $\mu(x)=\lca_S( \sigma(L(x)) )$ 
\end{itemize}
We call $\mu$ the reconciliation map from $(T,t,\sigma)$ to $S$. 
\end{definition}                

We note that $\mu^{-1}(\rho_S)=\emptyset$ holds as an immediate consequence of
property \textit{(v)}, which implies that no speciation node can be mapped
above $\lca_S(B)$, the unique child of $\rho_S$.

We illustrate this definition by means of an example in
Fig.~\ref{fig:mapping} and remark that it is consistent with the definition
of reconciliation maps for the case when the event labeling $t$ on $T$ is
not known \cite{Doyon:09}. Continuing with our notation from
Definition~\ref{def:mu} for the remainder of this section, we easily derive
their axiom set as
\ \newline 
\begin{lemma}
  If $\mu$ is a reconciliation map from $(T,t,\sigma)$ to $S$ 
  and $L$ is the leaf set of $T$ then, for all $x\in V$,
\begin{itemize}
  \item[(D1)]   $x\in L$ implies $\mu(x)=\sigma(x)$.
  \item[(D2.a)] $\mu(x)\in W$ implies $\mu(x)=\lca_S( \sigma(L(x)) )$.
  \item[(D2.b)] $\mu(x)\in H$ implies $\lca_S(\sigma(L(x))) \prec_S \mu(x)$.
  \item[(D3)]   Suppose $x,y\in V$ such that $x\prec_T y$. 
                If $\mu(x),\mu(y)\in H$ then 
                $\mu(x)\preceq_S \mu(y)$; otherwise 
                $\mu(x)\prec_S \mu(y)$.
\end{itemize}
\end{lemma}
\begin{proof} Suppose $x\in V$. Then
  (D1) is equivalent to \textit{(i)} and the fact that $t(x)=\odot$
  if and only if $x\in L$. Conditions \textit{(ii)} and \textit{(v)}
  together imply (D2.a). If $\mu(x)\in H$ then $x$ is duplication vertex
  of $T$.  From condition \textit{(iv)} we conclude that $\lca_S(\sigma(L(x)))
  \preceq_S \mu(x)$. Since $\lca_S( \sigma(L(x)) )\in W$, equality
  cannot hold and so (D2.b) follows.  (D3) is an immediate consequence of 
  \textit{(iv)}.
\end{proof}

\begin{figure}[tbp]
\begin{center}
\includegraphics[bb= 171 500 473 714, scale=0.9]{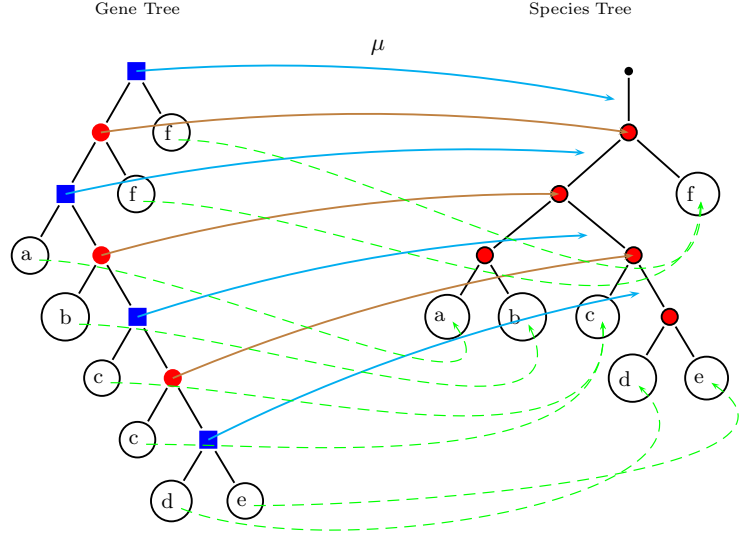}
\end{center}
\caption{Example of the mapping $\mu$ of nodes of the gene tree $T$ to the
  species tree $S$. Speciation nodes in the gene tree (red circles) are
  mapped to nodes in the species tree, duplication nodes (blue squares) are
  mapped to edges in the species tree.  $\sigma$ is shown as dashed green
  arrows. For clarity of exposition, we have identified the leaves of the 
  gene tree on the left with the species they reside in via the map 
  $\sigma$.}
\label{fig:mapping}
\end{figure}

For $T$ a gene tree, $B$ a set of species and maps $\sigma$ and $t$ as
above, our goal is now to characterize (1) those $(T,t,\sigma)$ for which a
species tree on $B$ exists and (2) species trees on $B$ that are species
trees for $(T,t,\sigma)$.

\section*{Results and Discussion}

\subsection*{Results}
Unless stated otherwise, we continue with our assumptions on $B$,
$(T,t,\sigma)$, and $S$ as stated in Definition~\ref{def:mu}. We start with
the simple observation that a reconciliation map from $(T,t,\sigma)$ to $S$
preserves the ancestor order of $T$ and hence $T$ imposes a strong
constraint on the relationship of most recent common ancestors in $S$:
\begin{lemma}\label{lem:1}
  Let $\mu:V\to W\cup H$ be a reconciliation map from $(T,t,\sigma)$ to
  $S$. Then
  \begin{equation}
    \lca_S(\mu(x),\mu(y)) \preceq_S\mu(\lca_T(x,y))
  \end{equation}
  holds for all $x,y\in V$. 
\end{lemma}
\begin{proof}
  Assume that $x$ and $y$ are distinct vertices of $T$.  Consider the
  unique path $P$ connecting $x$ with $y$. $P$ is uniquely subdivided into
  a path $P'$ from $x$ to $\lca_T(x,y)$ and a path $P''$ from $\lca_T(x,y)$
  to $y$. Condition (iv) implies that the images of the vertices of $P'$
  and $P''$ under $\mu$, resp., are ordered in $S$ with regards to
  $\preceq_S$ and hence are contained in the intervals $Q'$ and $Q''$ that
  connect $\mu(\lca_T(x,y))$ with $\mu(x)$ and $\mu(y)$, respectively.  In
  particular $\mu(\lca_T(x,y))$ is the largest element (w.r.t.\
  $\preceq_S$) in the union of $Q'\cup Q''$ which contains the unique path
  from $\mu(x)$ to $\mu(y)$ and hence also $\lca_S(\mu(x),\mu(y))$.
\end{proof}

Since a phylogenetic tree (in the original sense) $T$ is uniquely
determined by its induced triple set $\mathfrak{R}(T)$, it is reasonable to
expect that all the information on the species tree(s) for $(T,t,\sigma)$
is contained in the images of the triples in $\mathfrak{R}(T)$ (or more
precisely their leaves) under $\sigma$.  However, this is not the case in
general as the situation is complicated by the fact that not all triples in
$\mathfrak{R}(T)$ are informative about a species tree that displays
$T$. The reason is that duplications may generate distinct paralogs long
before the divergence of the species in which they eventually appear.  To
address this problem, we associate to $(T,t,\sigma)$ the set of triples
\begin{equation}
  \mathfrak{G} = \mathfrak{G}(T,t,\sigma) = 
  \left\{ r \in \mathfrak{R}(T)\big\vert 
    t(\lca_T(r))=\bullet \, \textrm{ and } \,
    \sigma(x)\not=\sigma(y),\, \textrm{for all}\,
    x,y\in L(r)\,\textrm{pairwise distinct}\right\}.
\end{equation}
As we shall see below, $\mathfrak{G}(T,t,\sigma)$ contains all the
information on a species tree for $(T,t,\sigma)$ that can be gleaned from
$(T,t,\sigma)$.

\begin{lemma}\label{lem:display}
  If $\mu$ is a reconciliation map
  from $(T,t,\sigma)$ to $S$  and $((x,y),z) \in \mathfrak{G}(T,t,\sigma)$ 
  then $S$ displays  $((\sigma(x),\sigma(y)),\sigma(z))$. 
\end{lemma}
\begin{proof}
Put $\mathfrak{G} = \mathfrak{G}(T,t,\sigma)$ and recall that
$L$ denotes the leaf set of $T$.  Let $\{x,y,z\}\in\binom{L}{3}$
and assume w.l.o.g.\, that $((x,y),z)\in \mathfrak G$. First 
consider the case that
  $t(\lca_T(x,y))=\bullet$. From condition \textit{(v)} we conclude that
  $\mu(\lca_T(x,y)) =\lca_S(\sigma(x),\sigma(y))$ and $\mu(\lca_T(x,y,z))
  =\lca_S(\sigma(x),\sigma(y),\sigma(z))$. Since, by assumption, 
  $\lca_T(x,y) \prec
  \lca_T(x,y,z)$, we have as a consequence of condition
  \textit{(iv)} that $\mu(\lca_T(x,y))\prec \mu(\lca_T(x,y,z))$.  From
  $\lca_T(x,z)=\lca_T(y,z) = \lca_T(x,y,z)$ we conclude that 
  $S$ must display $((\sigma(x),\sigma(y)),\sigma(z))$ as
  $S$ is assumed to be a species tree for $(T,t,\sigma)$.
  
  Now suppose that $t(\lca_T(x,y))=\square$ and therefore,
  $\mu(\lca_T(x,y))\in H$. Moreover, $\mu(\lca_T(x,y,z))\in W$ holds.
  Hence, Lemma~\ref{lem:1} and property \textit{(iv)} together 
  imply that $\lca_S(\sigma(x),\sigma(y))\prec_S\mu(\lca_T(x,y))\prec_S
  \mu(\lca_T(x,y,z))$. Thus, we again obtain that the triple 
  $((\sigma(x),\sigma(y)),\sigma(z))$ is displayed by $S$.
\end{proof}

\begin{figure}[tbp]	
\centering
  \includegraphics[bb=0 0 571 163, width=0.8\textwidth]{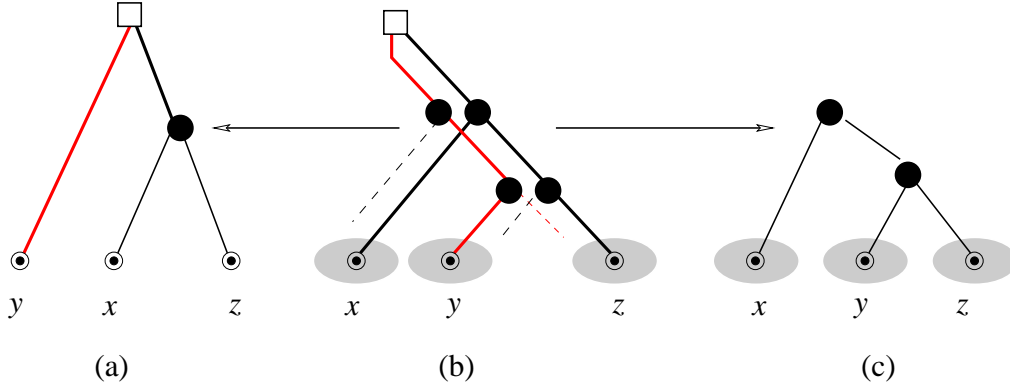}
  \caption{Triples from $T$ whose root is a duplication event are in
    general not displayed from the species tree $S$. (a) Triple with
    duplication event at the root obtained from the true evolutionary
    history of $T$ shown in panel (b). Panel (c) is the true species
    tree. In the triple (a) the species $y$ appears as the outgroup even
    though the $x$ is the outgroup in the true species tree.}
\label{fig:cnt1}
\end{figure}

It is important to note that a similar argument cannot be made for triples
in $\mathfrak{R}(T)$ rooted in a duplication vertex of $T$ as such triplets
are in general not displayed by a species tree for $(T,t,\sigma)$.  We
present the generic counterexample in Fig.~\ref{fig:cnt1}.

To state our main result (Theorem~\ref{thm:order}), we require a 
further definition.

\begin{definition}
  For $(T,t,\sigma)$, we define the set 
  \begin{equation}
    \mathfrak{S}=\mathfrak{S}(T,t,\sigma) = 
    \left\{ ((a,b),c)   |\, \exists ((x,y),z)\in\mathfrak{G}(T,t,\sigma) 
      \textrm{ with } 
      \sigma(x)=a,\, \sigma(y)=b,\,\textrm{and}\,\, \sigma(z)=c \right\}
  \end{equation}
\end{definition}

As an immediate consequence of Lemma~\ref{lem:display},
$\mathfrak{S}(T,t,\sigma)$ must be displayed by any species tree
for $(T,t,\sigma)$ with leaf set $B$.

\begin{theorem}
Let $S$ be a species tree with leaf set $B$. Then
there exists a reconciliation map $\mu$ from $(T,t,\sigma)$ to $S$
whenever $S$ displays all triples in $\mathfrak{S}(T,t,\sigma)$. 
\label{thm:order}
\end{theorem}

\begin{proof}
Recall that $L$ is the leaf set of $T=(V,E)$. Put
$S=(W,H)$ and $\mathfrak S=\mathfrak{S}(T,t,\sigma)$. 
We first consider the subset $G:=\{x\in V\mid t(x)\in\{\bullet,
  \odot\}\}$ of $V$ comprising of the leaves and speciation vertices of $T$.

  We explicitly construct the map $\mu:G\to W$ as follows. 
  For all $x\in V$, we put 
  \begin{itemize}
  \item[(M1)] $\mu(x) = \sigma(x)$ if $t(x)=\odot$,
  \item[(M2)] $\mu(x) = \lca_S(\sigma(L(x)))$ if $t(x)=\bullet$. 
  \end{itemize}
  Note that alternative (M1) ensures that $\mu$ satisfies Condition 
  \textit{(i)}. Also note that in view of the simple consequence following 
  the statement of Condition \textit{(C)} we have  
  for all $x\in V$ with $t(x)=\bullet$ that there are leaves $y',y''\in L(x)$
  with $\sigma(y')\ne\sigma(y'')$. Thus $\lca_S(\mu(L(x))\in W\setminus B$,
  i.e.\ $\mu$ satisfies Condition \textit{(ii)}. Also note that, by definition,
  Alternative (M2) ensures that $\mu$ satisfies Condition \textit{(v)}. 

  \par\noindent\textbf{Claim:} 
  If $x,y \in G$ with $x\prec_T y$ then $\mu(x)\prec_S \mu(y)$.
  \newline
  Since $y$ cannot be a leaf of $T$ as $x\prec_T y$ we have
  $t(y)=\bullet$. There are two cases to consider, either 
  $t(x)=\bullet$ or $t(x)=\odot$. In the latter case $\mu(x)=\sigma(x)\in
  B$ while $\mu(y)\in W\setminus B$ as argued above. Since $x\in L(y)$ we 
  have $\mu(x)\prec_S \mu(y)$, as desired. 
  \newline 
  Now suppose $t(x)=\bullet$. Again by the simple consequence following 
  Condition \textit{(C)}, there are leaves 
  $x', x'' \in L(x)$ with $a=\sigma(x')\neq \sigma(x'')=b$. Since 
  $x\prec_T y$ and $t(y)=\bullet$, by Condition \textit{(C)}, we conclude 
  that   $c = \sigma(y')\notin \sigma(L(x))$ holds
  for all $y'\in L(y)\setminus L(x)$. Thus, $((a,b),c)\in\mathfrak{S}$. 
  But then $((a,b),c)$ is displayed by $S$ and therefore 
  $\lca_S(a,b) \prec_S \lca_S(a,b,c)$. 
  Since this holds for all triples $((x', x''),y')\in \mathfrak{G}$ with 
  $x',x''\in L(x)$ and $y'\in L(y)\setminus L(x)$ we conclude 
  $\mu(x)=\lca_S(\sigma(L(x))) \prec_S 
  \lca_S(\sigma(L(x))\cup\sigma(L(y)\setminus L(x))) = 
  \lca_S(\sigma(L(y)))=\mu(y)$, establishing the Claim.

  It follows immediately that $\mu$ also satisfies Condition
  \textit{(iv.2)} if $x$ and $y$ are contained in $G$.

  Next, we extend the map $\mu$ to the entire vertex set $V$ of $T$ using
  the following observation.  Let $x \in V$ with $t(x)=\square$. We know by
  Lemma~\ref{lem:1} that $\mu(x)$ is an edge $[u,v]\in H$ so that
  $\lca_S(\sigma(L(x)))\preceq_S v$.  Such an edge exists for
  $v=\lca_S(\sigma(L(x)))$ by construction. Every speciation vertex $y\in
  V$ with $x\prec_T y$ therefore necessarily maps above this edge, i.e.,
  $u\preceq_S \mu(y)$ must hold. Thus we set
  \begin{itemize}
  \item[(M3)]$\mu(x)=[u,\lca_S(\sigma(L(x)))]$ if $t(x)=\square$.
  \end{itemize}
  which now makes $\mu$ a map from $V$ to $ W\cup H$.

  By construction, Conditions \textit{(iii)}, 
  \textit{(iv.2)} and \textit{(v)} are thus satisfied by $\mu$. 
  On the other hand, if there is speciation vertex $y$ between
  two duplication vertices $x$ and $x'$ of $T$, i.e., 
  $x\prec_T y \prec_T x'$, then $\mu(x)\prec_S\mu(x')$.
  Thus $\mu$  also satisfies Condition \textit{(iv.1)}.  
\newline
It follows that $\mu$ is a reconciliation map from $(T,t,\sigma)$ to $S$.
\end{proof}

\begin{corollary}
  Suppose that $S$ is a species tree for $(T,t,\sigma)$ and that $L$ and
  $B$ are the leaf sets of $T$ and $S$, respectively. Then a reconciliation
  map $\mu$ from $(T,t,\sigma)$ to $S$ can be constructed in 
  $O( |L| |B| )$.
\end{corollary}
\begin{proof}
  In order to find the image of an interior vertex $x$ of $T$ under $\mu$,
  it suffices to determine $\sigma(L(x))$ (which can be done for all $x$
  simultaneously e.g.\ by bottom up transversal of $T$ in $O(|L| |B|)$
  time) and $\lca_S(\sigma(L(x)))$. The latter task can be solved in linear
  time using the idea presented in \cite{Zhang:97} to calculate the lowest
  common ancestor for a group of nodes in the species tree.
\end{proof}

We remark that given a species tree $S$ on $B$ that displays all triples in
$\mathfrak{S}(T,t,\sigma)$, there is no freedom in the construction of a
reconciliation map on the set $\{x\in V\mid t(x)\in\{\bullet, \odot\}\}$.
The duplication vertices of $T$, however, can be placed differently,
resulting in possibly exponentially many reconciliation maps from
$(T,t,\sigma)$ to $S$.

Lemma~\ref{lem:display} implies that consistency of the triple set
$\mathfrak{S}(T,t,\sigma)$ is necessary for the existence of a
reconciliation map from $(T,t,\sigma)$ to a species tree on
$B$. Theorem~\ref{thm:order}, on the other hand, establishes that this is
also sufficient. Thus, we have
\begin{theorem}\label{theo:final}
  There is a species tree on $B$ for $(T,t,\sigma)$ if and only 
  if the triple set $\mathfrak{S}(T,t,\sigma)$ is consistent. 
\end{theorem}

We remark that a related result is proven in \cite[Theorem.5]{Chauve:09} for
the full tree reconciliation problem starting from a forest of gene trees.

\begin{figure}[tbp]
\centering
\includegraphics[bb= 110 480 569 753, width=\textwidth]{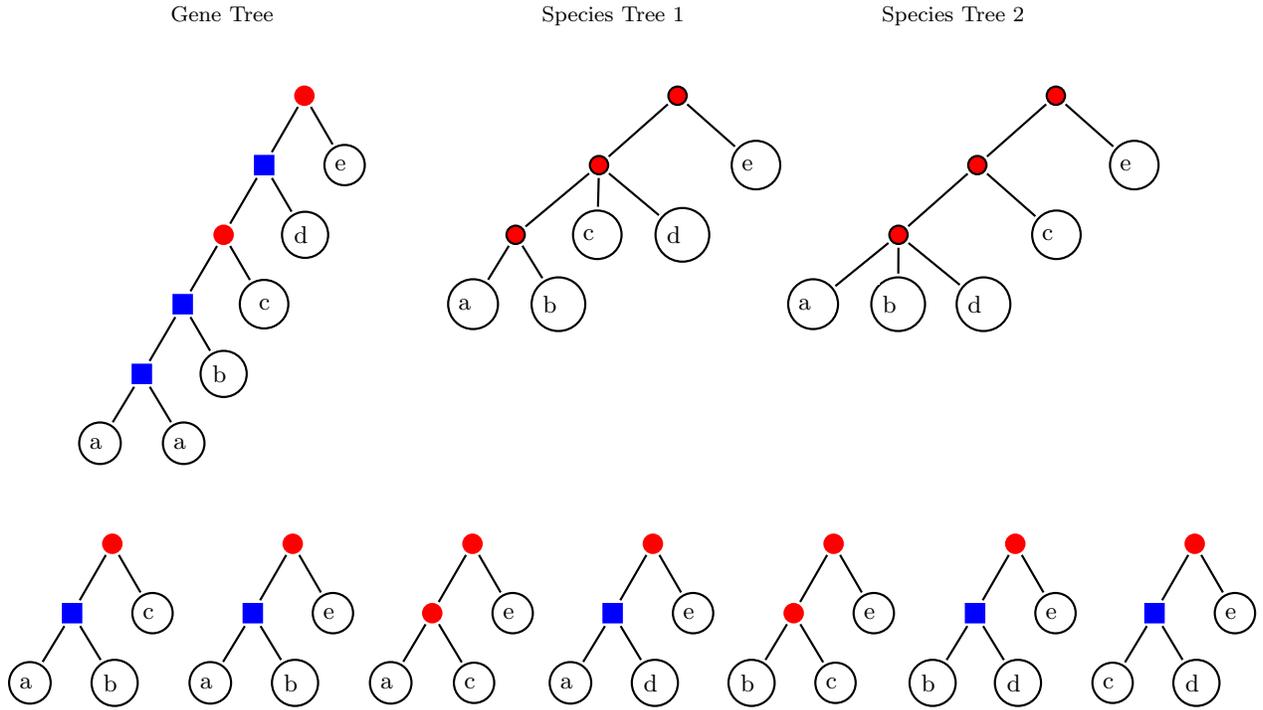}
\caption{The set $\mathfrak{S}(T,t,\sigma)$ inferred from the event labeled
  gene tree $(T,t,\sigma)$ does not necessarily define a unique species
  tree. For clarity of exposition, we have identified, via the map
  $\sigma$, the leaves of the gene tree and of the set of triples
  $\mathfrak{S}(T,t,\sigma)$ with the species they reside in .}
\label{fig:triples}
\end{figure}

It may be surprising that there are no strong restrictions on the set 
$\mathfrak{S}(T,t,\sigma)$ of triples that are implied by the fact that
they are derived from a gene tree $(T,t,\sigma)$.

\begin{theorem}
\label{thm:dings}
For every set $\mathfrak{X}$ of triples on some finite set $B$ of size
at least one there is a gene tree $T=(V,E)$ with leaf set $L$ together
with an event map $t:V\to\{\bullet,\square,\odot\} $ and a map
$\sigma:L\to B$ that assigns to every leaf of $T$ the species in $B$
it resides in such that $\mathfrak{X}=\mathfrak{S}(T,t,\sigma)$.
\end{theorem}
\begin{proof}
  Irrespective of whether $\mathfrak{X}$ is consistent or not we construct
  the components of the required 3-tuple $(T,t,\sigma)$ as follows: To
  each triple $r_k=((x_{k1},x_{k2}),x_{k3}) \in \mathfrak{X}$ we associate a
  triple $T_k=((a_{k1},a_{k2}),a_{k3})$ via a map $\sigma_k:L_k
  =\{a_{k1},a_{k2},a_{k3}\} \to \{x_{k1},x_{k2},x_{k3}\}$ with
  $\sigma(a_{ki})=x_{ki}$ for $i=1,2,3$ where we assume that for any two
  distinct triples $r_k,r_l\in\mathfrak X$ we have that $\sigma_k(L_k)\cap
  \sigma_l(L_l)=\emptyset$.  Then we obtain $T=(V, E)$ by first adding a
  single new vertex $\rho_T$ to the union of the vertex sets of the triples
  $ T_k$ and then connecting $\rho_T$ to the root $\rho_k$ of each of the
  triples $T_k$. Clearly, $T$ is a phylogenetic tree on
  $L=\stackrel{\cdot}{\bigcup}_{r_k\in \mathfrak X} L(\rho_k)$. Next, we
  define the map $t:V\to \{\bullet,\square,\odot\} $ by putting
  $t(\rho_T)=\square$, $t(a)=\odot$ for all $a\in L$ and $t(a)= \bullet$
  for all $a\in V-(L\cup\{\rho_T\})$. Finally, we define the map $\sigma:
  L\to B$ by putting, for all $a\in L$, $\sigma(a)= \sigma_k(a) $ where
  $a\in L_k$. Clearly $\mathfrak{S}(T,t,\sigma)=\mathfrak{X}$.
\end{proof}

We remark that the gene tree constructed in the proof of
Theorem~\ref{thm:dings} can be made into a binary tree by splitting the
root $\rho_T$ into a series of duplication and loss events so that each
subtree is the descendant of a different paralog.

Since by Theorem.~\ref{thm:dings} there are no restrictions on the possible
triple sets $\mathfrak{S}(T,t,\sigma)$, it is clear that $S$ will in
general not be unique. An example is shown in Fig.\ref{fig:triples}.

\subsubsection*{Results for simulated gene trees}

In order to determine empirically how much information on the species tree
we can hope to find in event labeled gene trees, we simulated species
trees together with corresponding event-labeled gene trees with different
duplication and loss rates. Approximately 150 species trees with 10 to 100
species were generated according to the the ``age model''
\cite{ageModel:11}. These trees are balanced and the edge lengths are
normalized so that the total length of the path from the root to each leaf
is 1.  For each species tree, we then simulated a gene tree as described in
\cite{Jobim:11}, with duplication and loss rate parameters $r\in[0,1]$
sampled uniformly. Events are modeled by a Poisson distribution with
parameter $r\cdot\ell$, where $\ell$ is the length of an edge as generated
by age model. Losses were additionally constrained to retain at least one
copy in each species, i.e., $\sigma(L)=B$ is enforced. After determining
the triple set $\mathfrak{S}(T,t,\sigma)$ according to
Theorem~\ref{thm:order}, we used \texttt{BUILD} \cite{Semple:book} (see
also \cite{Aho1981}) to compute the species tree. In all cases,
\texttt{BUILD} returns a tree that is a homomorphic contraction of the
simulated species tree. The difference between the original and the
reconstructed species tree is thus conveniently quantified as the
difference in the number of interior vertices.  Note that in our situation
this is the same as the split metric \cite{Semple:book}.

\begin{figure}[tbp]	
\centering
\includegraphics[width=\linewidth]{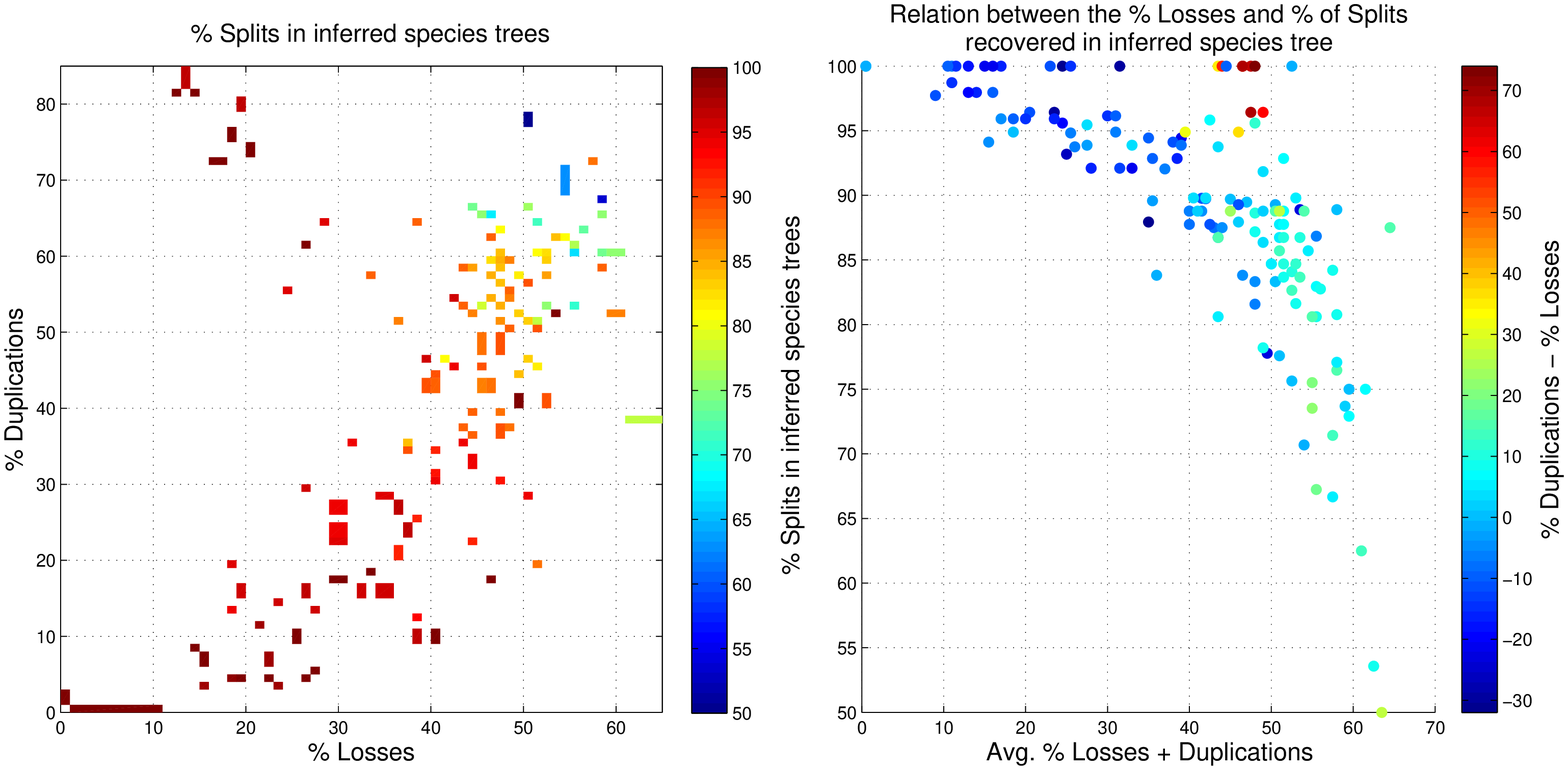}
\caption{\textbf{Left:} Heat map that represents the percentage of
  recovered splits in the inferred species tree from triples obtained from
  simulated event-labeled gene trees with different loss and duplication
  rates.  \newline \textbf{Right:} Scattergram that shows the average of
  losses and duplications in the generated data and the accuracy of the
  inferred species tree.}
\label{fig:splits}
\end{figure}

The results are summarized in Fig. \ref{fig:splits}. Not surprisingly, the
recoverable information decreases in particular with the rate of gene
loss. Nevertheless, at least 50\% of the splits in the species tree are
recoverable even at very high loss rates. For moderate loss rates, in
particular when gene losses are less frequent than gene duplications,
nearly the complete information on the species tree is preserved. It is
interesting to note that \texttt{BUILD} does not incorporate splits that
are not present in the input tree, although this is not mathematically
guaranteed.

\subsection*{Discussion}

Event-labeled gene trees can be obtained by combining the reconstruction of
gene phylogenies with methods for orthology detection. Orthology alone
already encapsulates partial information on the gene tree. More precisely,
the orthology relation is equivalent to a homomorphic image of the gene
tree in which adjacent vertices denote different types of events. We
discussed here the properties of reconciliation maps $\mu$ from a gene tree
$T$ along with an event labelling map $t$ and a gene to species assignment
map $\sigma$ to a species tree $S$ and show that $(T,t)$ event labeled gene
trees for which a species tree exists can be characterized in terms of the
set $\mathfrak{S}(T,t,\sigma)$ of triples that is easily constructed from a
subset of triples of $T$. Simulated data shows, furthermore, that such
trees convey a large amount of information on the underlying species tree,
even if the gene loss rate is high.

It can be expected for real-life data the tree $T$ contains errors so that
$\mathfrak{S} := \mathfrak{S}(T,t,\sigma)$ may not be consistent. In this
case, an approximation to the species tree could be obtained e.g.\ from a
maximum consistent subset of $\mathfrak{S}$. Although (the decision version
of) this problems is NP-complete \cite{Jansson2001,Wu2004}, there is a wide
variety of practically applicable algorithms for this task, see
\cite{He2006,Byrka:10}. Even if $\mathfrak{S}$ is consistent, the species
tree is usually not uniquely determined. Algorithms to list all trees
consistent with $\mathfrak{S}$ can be found e.g.\ in
\cite{Ng1996,Constantinescu1995}. A characterization of triple sets that
determine a unique tree can be found in \cite{Bryant1995}. Since our main
interest is to determine the constraints imposed by $(T,t,\sigma)$ on the
species tree $S$, we are interested in a least resolved tree $S$ that
displays all triples in $\mathfrak{S}$. The \texttt{BUILD} algorithm and
its relatives in general produce minor-minimal trees, but these are not
guaranteed to have the minimal number of interior nodes. Finding a species
tree with a minimal number of interior nodes is again a hard problem
\cite{Jansson:12}. At least, the vertex minimal trees are among the
possibly exponentially many minor minimal trees enumerated by Semple's
algorithms \cite{Semple:03}.

For a given species tree $S$, it is rather easy to find a reconciliation
map $\mu$ from $(T,t,\sigma)$ to $S$. A simple solution $\mu$ is closely related
to the so-called LCA reconcilation: every node $x$ of $T$ is mapped
to the last common ancestor of the species below it, $\lca_S(\sigma(L(x)))$
or to the edge immediately above it, depending on whether $x$ is speciation
or a duplication node. While this solution is unique for the speciation
nodes, alternative mappings are possible for the duplication nodes. The set
of possible reconciliation maps can still be very large despite the
specified event labels. 

If the event labeling $t$ is unknown, there is a reconciliation from any
gene tree $T$ to any species tree $S$, realized in particular by the LCA
reconciliation, see e.g.\ \cite{Chauve:09,Doyon:09}. The reconciliation
then defines the event types. Typically, a parsimony rule is then employed
to choose a reconciliation map in which the number of duplications and
losses is minimized, see e.g.\
\cite{Guigo1996,Bonizzoni2005,Burleigh2009,Gorecki2006}. In our setting, 
on the other hand, the event types are prescribed. This restricts the 
possible reconciliation maps so that the gene tree cannot be reconciled with an
arbitrary species tree any more. 

Since the observable events on the gene tree are fixed, the possible
reconciliations cannot differ in the number of duplications. Still, one may
be interested in reconciliation maps that minimize the number of loss
events. An alternative is to maximize the number of duplication events that
map to the same edge in $S$ to account for whole genome and chromosomal
duplication event \cite{Burleigh2009}. 

\section*{Conclusions}
Our approach to the reconciliation problem via event-labeled
gene trees opens up some interesting new avenues to understanding
orthology.  In particular, the results in this contribution combined with
those in \cite{Hellmuth:12d} concerning cographs should ultimately lead to
a method for automatically generating orthology relations that takes into
account species relationships without having to explicitly compute gene
trees. This is potentially very useful since gene tree estimation is one of
the weak points of most current approaches to orthology analysis.
    
\section*{Competing interests}
The authors declare that they have no competing interests.

\section*{Authors' contributions}

All authors contributed to the development of the theory.  MHR and NW
produced the simulated data. All authors contributed to writing, reading, and
approving the final manuscript.

\section*{Acknowledgements}

  This work was supported in part by the the \emph{Volkswagen Stiftung}
  (proj.\ no.\ I/82719) and the \emph{Deutsche Forschungsgemeinschaft}
  (SPP-1174 ``Deep Metazoan Phylogeny'', proj.\ nos.\ STA 850/2 and STA
  850/3).

 \bibliographystyle{plain}  
  \bibliography{species} 

\begin{thebibliography}{10}

\bibitem{Aho:81}
A.~V. Aho, Y.~Sagiv, T.~G. Szymanski, and J.~D. Ullman.
\newblock Inferring a tree from lowest common ancestors with an application to
  the optimization of relational expressions.
\newblock {\em SIAM J. Comput.}, 10:405--421, 1981.

\bibitem{Aho1981}
A.~V. Aho, Y.~Sagiv, T.~G. Szymanski, and J.~D. Ullman.
\newblock Inferring a tree from lowest common ancestors with an application to
  the optimization of relational expressions.
\newblock {\em SIAM J. Comput.}, 10:405--421, 1981.

\bibitem{Altenhoff:09}
A~M Altenhoff and C.~Dessimoz.
\newblock Phylogenetic and functional assessment of orthologs inference
  projects and methods.
\newblock {\em PLoS Comput Biol.}, 5:e1000262, 2009.

\bibitem{Arvestad2003}
L.~Arvestad, A.~C. Berglund, J.~Lagergren, and B.~Sennblad.
\newblock Bayesian gene/species tree reconciliation and orthology analysis
  using {MCMC}.
\newblock {\em Bioinformatics}, 19:i7--i15, 2003.

\bibitem{Bansal2008}
M.~S. Bansal and O.~Eulenstein.
\newblock The multiple gene duplication problem revisited.
\newblock {\em Bioinformatics}, 24:i132--i138, 2008.

\bibitem{Berglund:08}
A~C Berglund, E~Sj{\"o}lund, G~Ostlund, and E~L Sonnhammer.
\newblock \texttt{InParanoid 6}: eukaryotic ortholog clusters with inparalogs.
\newblock {\em Nucleic Acids Res.}, 36:D263--D266, 2008.

\bibitem{Bininda:book}
O.R.P Bininda-Emonds.
\newblock {\em Phylogenetic Supertrees}.
\newblock Kluwer Academic Press, Dordrecht, NL, 2004.

\bibitem{Boeckner:98}
Sebastian B{\"o}cker and Andreas W.~M. Dress.
\newblock Recovering symbolically dated, rooted trees from symbolic
  ultrametrics.
\newblock {\em Adv. Math.}, 138:105--125, 1998.

\bibitem{Bonizzoni2005}
Paola Bonizzoni, Gianluca Della~Vedova, and Riccardo Dondi.
\newblock Reconciling a gene tree to a species tree under the duplication cost
  model.
\newblock {\em Theor. Comp. Sci.}, 347:36--53, 2005.

\bibitem{Brandstaedt:99}
Andreas Brandst{\"a}dt, Van~Bang Le, and Jeremy~P Spinrad.
\newblock {\em Graph Classes: A Survey}.
\newblock SIAM Monographs on Discrete Mathematics and Applications. Soc. Ind.
  Appl. Math., Philadephia, 1999.

\bibitem{Bryant1995}
D.~Bryant and M.~Steel.
\newblock Extension operations on sets of leaf-labeled trees.
\newblock {\em Adv. Appl. Math.}, 16:425--453, 1995.

\bibitem{Burleigh2009}
J.~G. Burleigh, M.~S. Bansal, A.~Wehe, and O.~Eulenstein.
\newblock Locating large-scale gene duplication events through reconciled
  trees: implications for identifying ancient polyploidy events in plants.
\newblock {\em J. Comput. Biol.}, 16:1071--1083, 2009.

\bibitem{Byrka:10a}
J.~Byrka, P.~Gawrychowski, K.~T. Huber, and S.~Kelk.
\newblock Worst-case optimal approximation algorithms for maximizing triplet
  consistency within phylogenetic networks.
\newblock {\em J. Discr. Alg.}, 8:65--75, 2010.

\bibitem{Byrka:10}
Jaroslaw Byrka, Sylvain Guillemot, and Jesper Jansson.
\newblock New results on optimizing rooted triplets consistency.
\newblock {\em Discr. Appl. Math.}, 158:1136--1147, 2010.

\bibitem{Chauve2008}
C.~Chauve, J.~P. Doyon, and N.~El-Mabrouk.
\newblock Gene family evolution by duplication, speciation, and loss.
\newblock {\em J. Comput. Biol.}, 15:1043--1062, 2008.

\bibitem{Chauve:09}
Cedric Chauve and Nadia El-Mabrouk.
\newblock New perspectives on gene family evolution: Losses in reconciliation
  and a link with supertrees.
\newblock {\em LNCS}, 5541:46--58, 2009.

\bibitem{Constantinescu1995}
Mariana Constantinescu and David Sankoff.
\newblock An efficient algorithm for supertrees.
\newblock {\em J Classification}, 12:101--112, 1995.

\bibitem{Datta:09}
Ruchira~S. Datta, Christopher Meacham, Bushra Samad, Christoph Neyer, and
  Kimmen Sj{\"o}lander.
\newblock Berkeley {PHOG}: Phylofacts orthology group prediction web server.
\newblock {\em Nucl. Acids Res.}, 37:W84--W89, 2009.

\bibitem{Doyon:09}
Jean-Philippe Doyon, Cedric Chauve, and Sylvie Hamel.
\newblock Space of gene/species trees reconciliations and parsimonious models.
\newblock {\em J. Comp. Biol.}, 16:1399--1418, 2009.

\bibitem{Dress:book}
Andreas W.~M. Dress, Katharina~T. Huber, Jacobus Koolen, Vincent Moulton, and
  Andreas Spillner.
\newblock {\em Basic Phylogenetic Combinatorics}.
\newblock Cambridge University Press, Cambridge, 2011.

\bibitem{Fitch2000}
Walter~M. Fitch.
\newblock Homology: a personal view on some of the problems.
\newblock {\em Trends Genet.}, 16:227--231, 2000.

\bibitem{Goodstadt:06}
Leo Goodstadt and Chris~P Ponting.
\newblock Phylogenetic reconstruction of orthology, paralogy, and conserved
  synteny for dog and human.
\newblock {\em PLoS Comput. Biol.}, 2:e133, 2006.

\bibitem{Gorecki2006}
P.~G{\'o}recki and Tiuryn J.
\newblock {DSL}-trees: A model of evolutionary scenarios.
\newblock {\em Theor. Comp. Sci.}, 359:378--399, 2006.

\bibitem{Guigo1996}
R.~Guig{\'o}, I.~Muchnik, and T.~F. Smith.
\newblock Reconstruction of ancient molecular phylogeny.
\newblock {\em Mol. Phylogenet. Evol.}, 6:189--213, 1996.

\bibitem{Hahn2007}
M.~W. Hahn.
\newblock Bias in phylogenetic tree reconciliation methods: implications for
  vertebrate genome evolution.
\newblock {\em Genome Biol.}, 8:R141, 2007.

\bibitem{He2006}
Y.~J. He, T.~N. Huynh, J.~Jansson, and W.~K. Sung.
\newblock Inferring phylogenetic relationships avoiding forbidden rooted
  triplets.
\newblock {\em J Bioinform Comput Biol}, 4:59--74, 2006.

\bibitem{HSW:16}
M.~Hellmuth, P.F. Stadler, and N.~Wieseke.
\newblock The mathematics of xenology: Di-cographs, symbolic ultrametrics,
  2-structures and tree-representable systems of binary relations.
\newblock {\em J. Math. Biology}, 2016.
\newblock (in press) DOI: 10.1007/s00285-016-1084-3.

\bibitem{HLS+15}
M.~Hellmuth, N.~Wiesecke, H.P. Lenhof, M.~Middendorf, and P.F. Stadler.
\newblock Phylogenomics with paralogs.
\newblock {\em Proceedings of the National Academy of Sciences (PNAS)},
  112(7):2058--2063, 2015.

\bibitem{HW:15}
M.~Hellmuth and N.~Wieseke.
\newblock On symbolic ultrametrics, cotree representations, and cograph edge
  decompositions and partitions.
\newblock In {\em Computing and Combinatorics: 21st International Conference
  (COCOON), 2015}, pages 609--623, Cham, 2015. Springer International
  Publishing.

\bibitem{HW:16book}
M.~Hellmuth and N.~Wieseke.
\newblock From sequence data incl. orthologs, paralogs, and xenologs to gene
  and species trees.
\newblock In Pontarotti Pierre, editor, {\em Evolutionary Biology}, pages
  373--392, Cham, 2016. Springer International Publishing.

\bibitem{HW:16a}
M.~Hellmuth and N.~Wieseke.
\newblock On tree representations of relations and graphs: Symbolic
  ultrametrics and cograph edge decompositions.
\newblock {\em J. Comb. Opt.}, 2017.
\newblock (in press) DOI 10.1007/s10878-017-0111-7.

\bibitem{Hellmuth:12d}
Marc Hellmuth, Maribel Hernandez-Rosales, Katharina~T. Huber, Vincent Moulton,
  Peter~F. Stadler, and Nicolas Wieseke.
\newblock Orthology relations, symbolic ultrametrics, and cographs.
\newblock {\em J. Math. Biol.}, 2012.
\newblock doi: 10.1007/s00285-012-0525-x.

\bibitem{Jobim:11}
Maribel Hernandez-Rosales, Nicolas Wieseke, Marc Hellmuth, and Peter~F.
  Stadler.
\newblock Simulation of gene family histories.
\newblock Technical Report 12-017, Univ.\ Leipzig, 2011.

\bibitem{Hubbard:07}
T~J Hubbard, B~L Aken, K~Beal, B~Ballester, M~Caccamo, Y~Chen, L~Clarke,
  G~Coates, F~Cunningham, T~Cutts, T~Down, S~C Dyer, S~Fitzgerald,
  J~Fernandez-Banet, S~Graf, S~Haider, M~Hammond, J~Herrero, R~Holland, K~Howe,
  K~Howe, N~Johnson, A~Kahari, D~Keefe, F~Kokocinski, E~Kulesha, D~Lawson,
  I~Longden, C~Melsopp, K~Megy, P~Meidl, B~Ouverdin, A~Parker, A~Prlic, S~Rice,
  D~Rios, M~Schuster, I~Sealy, J~Severin, G~Slater, D~Smedley, G~Spudich,
  S~Trevanion, A~Vilella, J~Vogel, S~White, M~Wood, T~Cox, V~Curwen, R~Durbin,
  X~M Fernandez-Suarez, P~Flicek, A~Kasprzyk, G~Proctor, S~Searle, J~Smith,
  A~Ureta-Vidal, and E~Birney.
\newblock Ensembl 2007.
\newblock {\em Nucleic Acids Res}, 35:D610--617, 2007.

\bibitem{Jansson:05}
J~Jansson, J.~H.-K. Ng, K.~Sadakane, and W.-K. Sung.
\newblock Rooted maximum agreement supertrees.
\newblock {\em Algorithmica}, 43:293--307, 2005.

\bibitem{Jansson2001}
Jesper Jansson.
\newblock On the complexity of inferring rooted evolutionary trees.
\newblock {\em Electronic Notes Discr. Math.}, 7:50--53, 2001.

\bibitem{Jansson:12}
Jesper Jansson, Richard~S. Lemence, and Andrzej Lingas.
\newblock The complexity of inferring a minimally resolved phylogenetic
  supertree.
\newblock {\em SIAM J. Comput.}, 41:272--291, 2012.

\bibitem{ageModel:11}
S.~Keller-Schmidt, M.~Tu{\u{g}}rul, V.~M. Egu{\'{\i}}luz,
  E.~Hern{\'a}ndez-Garc{\'{\i}}i, and K.~Klemm.
\newblock An age dependent branching model for macroevolution.
\newblock Technical Report 1012.3298v1, arXiv, 2010.

\bibitem{Larget2010}
B.~R. Larget, S.~K. Kotha, C.~N. Dewey, and C.~Ane.
\newblock {BUCKy}: gene tree/species tree reconciliation with {Bayesian}
  concordance analysis.
\newblock {\em Bioinformatics}, 26:2910--2911, 2010.

\bibitem{Lechner:11a}
Marcus Lechner, Sven Findei{\ss}, Lydia Steiner, Manja Marz, Peter~F. Stadler,
  and Sonja~J. Prohaska.
\newblock \texttt{Proteinortho:} detection of (co-)orthologs in large-scale
  analysis.
\newblock {\em BMC Bioinformatics}, 12:124, 2011.

\bibitem{Li:06}
H~Li, A~Coghlan, J~Ruan, L~J Coin, J~K H{\'e}rich{\'e}, L~Osmotherly, R~Li,
  T~Liu, Z~Zhang, L~Bolund, G~K Wong, W~Zheng, P~Dehal, J~Wang, and R~Durbin.
\newblock \texttt{TreeFam}: a curated database of phylogenetic trees of animal
  gene families.
\newblock {\em Nucleic Acids Res.}, 34:D572--D580, 2006.

\bibitem{Li:03}
Li~Li, Christian~J Stoeckert, and David~S Roos.
\newblock Orthomcl: identification of ortholog groups for eukaryotic genomes.
\newblock {\em Genome research}, 13(9):2178--2189, 2003.

\bibitem{Ng1996}
Meei~Pyng Ng and Nicholas~C. Wormald.
\newblock Reconstruction of rooted trees from subtrees.
\newblock {\em Discr. Appl. Math.}, 69:19--31, 1996.

\bibitem{Page1997}
R.~D. Page and M.~A. Charleston.
\newblock From gene to organismal phylogeny: reconciled trees and the gene
  tree/species tree problem.
\newblock {\em Mol. Phylogenet. Evol.}, 7:231--240, 1997.

\bibitem{Pryszcz:11}
Leszek~P. Pryszcz, Jaime Huerta-Cepas, and Toni Gabald{\'o}n.
\newblock \texttt{MetaPhOrs}: orthology and paralogy predictions from multiple
  phylogenetic evidence using a consistency-based confidence score.
\newblock {\em Nucleic Acids Res.}, 39:e32, 2011.

\bibitem{Henzinger:99}
Monika Rauch~Henzinger, Valerie King, and Tandy Warnow.
\newblock Constructing a tree from homeomorphic subtrees, with applications to
  computational evolutionary biology.
\newblock {\em Algorithmica}, 24:1--13, 1999.

\bibitem{Semple:03}
Charles Semple.
\newblock Reconstructing minimal rooted trees.
\newblock {\em Discr. Appl. Math}, 127:489--503, 2003.

\bibitem{Semple:book}
Charles Semple and Mike Steel.
\newblock {\em Phylogenetics}, volume~24 of {\em Oxford Lecture Series in
  Mathematics and its Applications}.
\newblock Oxford University Press, Oxford, UK, 2003.

\bibitem{Tatusov:00}
R~L Tatusov, M~Y Galperin, D~A Natale, and E~V Koonin.
\newblock The {COG} database: a tool for genome-scale analysis of protein
  functions and evolution.
\newblock {\em Nucleic Acids Res}, 28:33--36, 2000.

\bibitem{vanIersel:09}
L.~van Iersel, S.~Kelk, and M.~Mnich.
\newblock Uniqueness, intractability and exact algorithms: reflections on
  level-$k$ phylogenetic networks.
\newblock {\em J. Bioinf. Comp. Biol.}, 7:597--623, 2009.

\bibitem{Wheeler:08}
D~L Wheeler, T~Barrett, D~A Benson, S~H Bryant, K~Canese, V~Chetvernin, D~M
  Church, M~Dicuccio, R~Edgar, S~Federhen, M~Feolo, L~Y Geer, W~Helmberg,
  Y~Kapustin, O~Khovayko, D~Landsman, D~J Lipman, T~L Madden, D~R Maglott,
  V~Miller, J~Ostell, K~D Pruitt, G~D Schuler, M~Shumway, E~Sequeira, S~T
  Sherry, K~Sirotkin, A~Souvorov, G~Starchenko, R~L Tatusov, T~A Tatusova,
  L~Wagner, and E~Yaschenko.
\newblock Database resources of the national center for biotechnology
  information.
\newblock {\em Nucleic Acids Res.}, 36:D13--D21, 2008.

\bibitem{Wu2004}
Bang~Ye Wu.
\newblock Constructing the maximum consensus tree from rooted triples.
\newblock {\em J. Comb. Optimization}, 8:29--39, 2004.

\bibitem{Zhang:97}
L.~Zhang.
\newblock On a {M}irkin-{M}uchnik-{S}mith conjecture for comparing molecular
  phylogenies.
\newblock {\em J. Comput. Biol.}, 4:177--187, 1997.

\end{thebibliography}


\end{document}